\documentclass[12pt]{amsart}
\usepackage[utf8]{inputenc}
\usepackage{amsmath, amsthm,amssymb,bbm,enumerate,mathrsfs,mathtools}
\usepackage{a4wide}
\usepackage{tikz,xcolor,verbatim}
\usetikzlibrary{calc,shapes}
\usepackage{hyperref}
\usepackage{makecell}
\usepackage[numbers,sort&compress]{natbib}

\title[Square roots of nearly planar graphs]{Square roots of nearly planar graphs}
\thanks{Supported by the ERC-CZ project LL2005 (Algorithms and complexity within and beyond bounded expansion) of the Ministry of Education of Czech Republic.}

\author[Dvo\v{r}\'{a}k]{Zden\v{e}k Dvo\v{r}\'{a}k}
\author[Lahiri]{Abhiruk Lahiri}
\author[Moore]{Benjamin Moore}
\address{Computer Science Institute, Charles University, Prague, Czech Republic}
\email{\{rakdver,lahiri,brmoore\}@iuuk.mff.cuni.cz}

\date{}

\newtheorem{thm}[equation]{Theorem}

\newtheorem{conj}[equation]{Conjecture}
\newtheorem{cor}[equation]{Corollary}

\theoremstyle{definition}

\newtheorem{obs}[equation]{Observation}

\newtheoremstyle{case}{}{}{\normalfont}{}{\itshape}{\normalfont:}{ }{}

\theoremstyle{case}

\numberwithin{equation}{section}

\usepackage{enumitem}
\newcommand{\decprob}[3]{
\smallskip\noindent\textbf{#1}
\vspace{-0.5\topsep}
\begin{description}[labelindent=\parindent]
 \setlength\itemsep{0em}
 \item[Instance] #2
 \item[Question] #3
\end{description}
}

\begin{document}

\maketitle

\begin{abstract}
We prove that it is NP-hard to decide whether a graph is the square of a 6-apex graph.
This shows that the square root problem is not tractable for squares of sparse graphs
(or even graphs from proper minor-closed classes).
\end{abstract}

\section{Introduction}
Given a graph $G$, the \textit{square} of $G$, denoted $G^{2}$, is the graph where $V(G^{2}) = V(G)$, and $uv \in E(G^{2})$ if and only if the distance between $u$ and $v$ in $G$ is at most $2$. For a graph $G$, a \textit{square root} of $G$ is a graph $H$ such that $H^{2} = G$. Note that a graph may have many possible square roots, for example, $K_{5}$ has both $C_{5}$ and $K_{1,4}$ as square roots. 

Let $\mathcal{H}$ be a fixed class of graphs. The problem we are interested in is:

 \decprob{$\mathcal{H}$-square-root}{A graph $G$.}{Is there a square root $H$ of $G$ such that $H \in \mathcal{H}$?}

There are some classes for which $\mathcal{H}$-square-root is known to be in P. For instance, if $\mathcal{H}$ is the class of bipartite graphs \cite{bipartiteroots}, outerplanar graphs \cite{Outerplanarroots}, proper interval graphs \cite{splitandchords} or the class of graphs with girth at least 6 \cite{girth6}, then $\mathcal{H}$-square-root is in P. Further, it is possible to compute some cut vertices of square roots given the square, \cite{findingcutvertices}, so the $\mathcal{H}$ problem is in P for classes of graphs such as trees, cacti, and block graphs.  On the other hand, the problem is known to be NP-complete when $\mathcal{H}$ is the class of all graphs \cite{MOTWANI} (despite a characterization of when a graph has a square root \cite{squarecharacterization}),  graphs with girth at least $5$ \cite{girth5}, chordal graphs \cite{splitandchords} and split graphs \cite{splitandchords}. 

Observing that the hard graph classes above are relatively dense, it has been asked multiple times if the $\mathcal{H}$-square-root problem is in P for sparse graph classes (for instance, see \cite{Outerplanarroots}). We show here that the answer is no. Recall that a graph $G$ is \textit{$k$-apex} if there exists a set of at most $k$ vertices such that the deletion of these vertices results in a planar graph. Using a mild tweak on the reduction in \cite{splitandchords} (showing the hardness of finding chordal square roots), we prove the following theorem: 

\begin{thm}
\label{mainthm}
The $\mathcal{H}$-square-root problem is NP-complete for the class $\mathcal{H}$ of $6$-apex graphs. 
\end{thm}

Actually, we prove the following slightly stronger result, showing the hardness of the following version of the problem.

\decprob{Promise-$\mathcal{H}$ square-root}{Graph $G$ that either has no square root, or has a square root belonging to~$\mathcal{H}$.}{Does $G$ have a square root?}

The difference here is that we are allowed not to deal with the graphs that have a square root, but no square root belonging to $\mathcal{H}$
(whereas $\mathcal{H}$-square-root has to answer NO for such instances).

\begin{thm}
\label{mainthm-prom}
The promise-$\mathcal{H}$ square-root problem is NP-complete for the class $\mathcal{H}$ of $6$-apex graphs. 
\end{thm}

Clearly, Theorem~\ref{mainthm-prom} implies Theorem~\ref{mainthm}.
Let us remark that the class of $6$-apex graph is minor-closed.  As such, it lies at the bottom of the bounded-expansion hierarchy of
sparse graph classes~\cite{nesbook}.  Still, it is natural to ask whether the apex vertices are indeed necessary, or whether
already the planar square-root problem is hard.  We believe the latter is the case.

\begin{conj}
The $\mathcal{H}$-square-root problem is NP-complete for the class $\mathcal{H}$ of planar graphs.
\end{conj}

\section{Set Splitting Preliminaries}

Before proving Theorem \ref{mainthm}, we first show a variant of set splitting is NP-complete. Let $S$ be a set and $\mathcal{C}$ a collection of subsets of $S$. The \textit{incidence graph} of $S$ and $C$ is the bipartite graph $I$ with bipartition $(A,B)$ such that for every element $e \in S$, we have a vertex $x_{e} \in A$, and for every set $c \in \mathcal{C}$, we have a vertex $v_{c} \in B$, and $x_{e}v_{c} \in E(I)$ if and only if $e \in c$. We are interested in the following problem:

\medskip

\decprob{Three parts planar set splitting}{A set $S$,  a collection $\mathcal{C}$ of subsets of $S$ where all subsets have at least three elements, and such that the incidence graph of $S$ and $\mathcal{C}$ is planar.}{Is there a partition of $S$ into three sets $S_{1},S_{2},S_{3}$ such that for all $c \in \mathcal{C}$ there is an element of $c$ in $S_{i}$ for each $i \in \{1,2,3\}$?}

\medskip

We first prove this problem is NP-complete by a reduction from reducing planar $3$-colouring, which is well-known to be NP-complete~\cite{garey}.

\begin{obs}
The three parts planar set splitting problem is NP-complete. 
\end{obs}

\begin{proof}
Let $G$ be an instance of planar $3$-colouring. Let $G'$ be the graph obtained from $G$ by taking every edge $e=xy$ and adding a new vertex $z_{xy}$ adjacent to both $x$ and $y$. Observe that $G'$ is planar, and that $G'$ is $3$-colourable if and only if $G$ is $3$-colourable. Further note that for any new vertex $z_{xy}$, in any $3$-colouring of $G'$, all three colours appear on the triangle $x,y,z_{xy}$. 

Now let $S$ be the set of vertices of $G'$, and let $\mathcal{C}$ be the collection of triangles which contain a vertex in $V(G') -V(G)$. Notice that the incidence graph of $S$ and $C$ is planar, since for any vertex in $V(G') -V(G)$, this vertex corresponds to a vertex of degree one in the incidence graph, and after removing all such vertices, we simply end up with a graph that is a subdivision of $G$. We claim that $G'$ (and hence $G$) has a $3$-colouring if and only if there is a partition of $S$ into three sets $S_{1},S_{2},S_{3}$ such that for each set $c \in \mathcal{C}$, there is an element of $c$ which belongs to $S_{i}$ for $i \in \{1,2,3\}$.  

To see this, suppose $G'$ has a $3$-colouring $f$. Let $S_{i}$ for $i \in \{1,2,3\}$ be the set of vertices which get colour $i$ in $f$. This is a partition of $S$, and for any triangle $T$ containing a vertex in $V(G')-V(G)$, $f$ colours the triangle with three distinct colours, hence we get a solution to the three parts planar set splitting problem.  Conversely, if $S_{1},S_{2},S_{3}$ is a partition of $S$ such that for any set $c \in \mathcal{C}$, there is an element of $c$ in each of $S_{1},S_{2},S_{3}$, then simply colour the vertices of $G'$ with colour $i$ if the vertex lies in $S_{i}$.
Since every edge of $G'$ is contained in a triangle belonging to $\mathcal{C}$, this is a proper $3$-colouring of $G'$, completing the reduction. 
\end{proof}

\section{The reduction}

Before jumping into the reduction, we first recall the ``tail'' structure which is frequently used in square root complexity results (see, for example, \cite{MOTWANI, splitandchords,girth5}). 

Let $G$ be any graph.  Let $N_{G}(v)$ denote the set of neighbors a vertex $v$ in $G$.
An \emph{$X$-tail} at a vertex $v \in V(G)$ consists of three vertices $v_{1},v_{2},v_{3} \in V(G)$
such that $N_G(v)\neq\{v_2,v_3\}$, $N_G(v_1)=\{v_2,v_3\}$, $N_G(v_2)=\{v_1, v_3, v\}$, and $N_G(v_3)=\{v_1,v_2,v\}\cup X$, where
$X$ is a subset of $N_G(v)\setminus \{v_2,v_3\}$.
The following observation is well known and illustrates why tails are useful: 

\begin{obs}[\cite{MOTWANI}]
\label{tail}
If $G$ is a graph and $v_1$, $v_2$, $v_3$ is an $X$-tail at $v$, then for any square root $H$ of $G$,
we have that $N_{H}(v_{1}) = \{v_{2}\}$, $N_{H}(v_{2}) = \{v_{1},v_{3}\}$,
$N_{H}(v_{3}) = \{v,v_{2}\}$, and $N_{H}(v) = \{v_3\}\cup X$.
\end{obs}

Now we give the reduction from the three part planar set splitting problem, which we have shown to be NP-complete in the previous section.

Let $S, \mathcal{C}$ be an instance of the three part planar set splitting problem, with $\mathcal{C}\neq\emptyset$. We construct a graph $G$ with the following vertex set:

\begin{itemize}
    \item A vertex $x_{s}$ for each element $s \in S$,
    \item For each $c \in \mathcal{C}$, we add a vertex $x_{c}$, as well as three additional vertices $x_{c}^{1},x_{c}^{2},x_{c}^{3}$,
    \item Vertices $a_i$, $b_i$, $b_i^1$, $b_i^2$, $b_i^3$ for $i\in\{1,2,3\}$.
\end{itemize}
And edge set:
\begin{itemize}
    \item[(i)] For any two distinct elements $s,s' \in S$, we add the edge $x_{s}x_{s'}$.
    \item[(ii)] For $i \in \{1,2,3\}$, and any element $s \in S$, we add the edges $a_{i}x_{s}$ and $b_ix_s$.
    \item[(iii)] For $i \in \{1,2,3\}$ and for any set $c \in \mathcal{C}$, we add the edges $a_{i}x_{c}$ and $b_ix_c$.
    \item[(iv)] For any two distinct sets $c_{1},c_{2} \in \mathcal{C}$, if $c_{1} \cap c_{2} \neq \emptyset$ then we add the edge $x_{c_{1}}x_{c_{2}}$.
    \item[(v)] For each $c \in \mathcal{C}$, add edges $x_{c}^{1}x_{c}^{2}, x_{c}^{1}x_{c}^{3},x_{c}^{2}x_{c}^{3},x_{c}^{2}x_{c},x_{c}^{3}x_{c}$. Further, for each element $s \in c$, add the edges $x_{c}^{3}x_{s}$ and $x_{c}x_{s}$. 
    Hence, $x_{c}^{1}$, $x_{c}^{2}$, $x_{c}^{3}$ is an $\{x_s:s\in c\}$-tail at $x_c$.
    \item[(vi)] For $i,j \in \{1,2,3\}$, we add the edges $a_{i}b_j$, and the edges $b_ib_j$ if $i<j$.
    \item[(vii)] For $i\in \{1,2,3\}$, we add the edges
    $b_i^{1}b_i^{2}, b_i^{1}b_i^{3}, b_i^{2}b_i^{3}, b_i^{2}b_i, b_i^{3}b_i$,
    for every element $s \in S$, add the edge $b_i^{3}x_{s}$, and add the edge $b_i^{3}a_{i}$.
    Hence, $b_i^{1}$, $b_i^{2}$, $b_i^{3}$ is an $(\{x_s:s\in S\}\cup \{a_i\})$-tail at $b_i$.
\end{itemize}

We claim that $G$ has a square root $H$ if and only if $S, \mathcal{C}$ is a YES-instance of the three part planar set splitting problem,
and if it has one, $H$ is $6$-apex.
We collect some basic observations about any square root of $G$ (if it exists).

\begin{obs}
\label{existanceofedge}
Let $H$ be a square root of $G$. For each $c \in \mathcal{C}$ and $i \in \{1,2,3\}$, there is $s\in c$ such that $x_sa_{i} \in E(H)$. 
\end{obs}

\begin{proof}
Because of the tail at $x_c$, Observation~\ref{tail} implies that $N_H(x_c)=\{x_c^3\}\cup \{x_s:s\in c\}$ and
that $x_c^3$ has no neighbors outside of the tail.  Since $G$ has the edge $x_{c}a_{i}$, $x_{c}$ must be at distance two from $a_{i}$,
and this implies that there is an edge $x_sa_{i}$ for some $s\in c$. 
\end{proof}

\begin{obs}
\label{exactlyone}
Let $H$ be a square root of $G$. For any element $s \in S$, the vertex $x_{s}$ is adjacent to at most one of $a_{i}$ for $i \in \{1,2,3\}$ in $H$.
\end{obs}

\begin{proof}
Since $\{a_1,a_2,a_3\}$ is an independent set in $G$, no two of these vertices can have a common neighbor in $H$.
\end{proof}

Combining the previous two observations, we obtain one implication of the reduction.
\begin{cor}\label{valid}
If $G$ has a square root, then $S, \mathcal{C}$ is a YES-instance of the three part planar set splitting problem.
\end{cor}
\begin{proof}
Suppose that $H$ is a square root of $G$.  For $i\in\{1,2,3\}$, let $S'_i=\{s\in S:x_sa_i\in E(H)\}$.  By Observation~\ref{exactlyone},
these sets are pairwise disjoint.  Hence, $S_1=S'_1$, $S_2=S'_2$, and $S_3=S\setminus (S'_1\cup S'_2)$ is a partition of $S$
with $S'_3\subseteq S_3$.
Moreover, by Observation~\ref{existanceofedge}, for each $c\in \mathcal{C}$ and $i\in \{1,2,3\}$, we have $c\cap S_i\neq\emptyset$.
\end{proof}

We now need to show the converse.

\begin{obs}\label{back}
If $S,\mathcal{C}$ is a YES-instance of the three planar set splitting problem, then $G$ has a $6$-apex square root.
\end{obs}
\begin{proof}
Let $S_{1},S_{2},S_{3}$ be a partition of $S$ such that for each $c \in \mathcal{C}$ and $i\in\{1,2,3\}$,
we have that there is an element $s_{c,i} \in c\cap S_{i}$.  Let us also choose $s_{c,4}\in c$ arbitrarily.
Let $H$ be the graph with $V(H)=V(G)$ and with the following edges:
\begin{itemize}
\item The edges forced by the tails, that is, for each $c\in \mathcal{C}$,
the edges $x_c^1x_c^2$, $x_c^2x_c^3$, $x_c^3x_c$, and $x_cx_s$ for each $s\in c$;
and for $i\in\{1,2,3\}$, the edges $b_i^1b_i^2$, $b_i^2b_i^3$, $b_i^3b_i$, $b_ix_s$ for each
$s\in S$ and $b_ia_i$ for $i\in\{1,2,3\}$.
\item The edges $x_sa_i$ for each $i\in\{1,2,3\}$ and $s\in S_i$.
\end{itemize}
Note that $H-\{a_1,a_2,a_3, b_1,b_2,b_3\}$ is planar, as it is obtained from the incidence graph of $S,\mathcal{C}$
by adding pendant paths corresponding to the tails.  Hence, $H$ is $6$-apex.

Now, $H^2$ contains the edges (i), since $x_s$ and $x_{s'}$ are both adjacent to $b_1$ in $H$.
It contains the edges (ii), since $x_sb_i,a_ib_i\in E(H)$.
It contains the edges (iii), since $a_ix_{s_{c,i}}, b_ix_{s_{c,i}},x_{s_{c,i}}x_c\in E(H)$.
It contains the edges (iv), since if $s\in c_1\cap c_2$, then $x_{c_1}x_s,x_{c_2}x_s\in E(H)$.
It contains the edges (v), since these follow from the path $x_{c}^{1}x_{c}^{2}x_{c}^{3}x_c$
and the edges $x_cx_s$ for $s\in c$.
It contains the edges (vi), since for any $c\in \mathcal{C}$,
$a_ix_{s_{c,i}},b_jx_{s_{c,i}}, b_ix_{s_{c,i}}\in E(H)$.
Finally, it contains the edges (vii), since these follow from the path $b_i^{1}b_i^{2}b_i^{3}b_i$ and
the edges $b_ia_i$ and $b_ix_s$ for $s\in S$.  Therefore, $G\subseteq H^2$.

To see that $G=H^2$, it suffices to show that the neighborhood of each vertex of $H$ induces a clique in $G$.
For $i\in\{1,2,3\}$,
\begin{itemize}
\item $N_H(a_i)=S_i\cup \{b_i\}$ is a clique in $G$ covered by the edges (i) and (ii);
\item $N_H(b_i)=S\cup\{a_i,b_i^3\}$ is a clique in $G$ covered by the edges (i), (ii), and (vii);
\item for $j\in\{1,2,3\}$, $N_H(b_i^j)$ is a clique in $G$ of size at most two covered by the edges (vii);
\item for $s\in S_i$, $N_H(x_s)=\{x_c:c\in \mathcal{C}, s\in c\}\cup\{a_i,b_1,b_2,b_3\}$
is a clique in $G$ covered by the edges (iii), (iv), and (vi).
\end{itemize}
Moreover, for $c\in\mathcal{C}$, $N_H(x_c)=\{x_s:s\in c\}\cup \{x_c^3\}$ is a clique in $G$ covered by the edges (i) and (v),
and for $j\in\{1,2,3\}$, $N_H(x_c^j)$ is clique in $G$ of size at most two covered by the edges (v).
\end{proof}

We can now prove our main result.
\begin{proof}[Proof of Theorem~\ref{mainthm-prom}]
Clearly, promise-$\mathcal{H}$ square-root problem is in NP, since we can guess a square root $H$ and easily check if $H^{2} = G$.

Given an instance $S,\mathcal{C}$ of the three planar set splitting problem, we create the graph $G$ as described at the beginning
of the section; clearly, this can be done in polynomial time.  By Corollary~\ref{valid}, if $S,\mathcal{C}$ is a NO-instance,
then $G$ has no square root; and by Observation~\ref{back}, if it is a YES-instance, then $G$ has a 6-apex square root.
Hence, this is a polynomial-time reduction from the the three planar set splitting problem to the promise-$\mathcal{H}$ square-root problem,
showing the NP-completeness of the latter.
\end{proof}

\bibliography{squarebib}
 \bibliographystyle{plain}
\end{document}